\documentclass[a4paper]{article}
\usepackage[latin1]{inputenc}
\usepackage[T1]{fontenc}
\usepackage[english]{babel}
\usepackage{amsfonts}
\usepackage{amssymb}
\usepackage{amsmath}
\usepackage{amsthm}
\usepackage{graphicx}
\usepackage{color}
\usepackage{latexsym}
\usepackage{empheq}

\newtheorem{theorem}{Theorem}

\newtheorem{remark}{Remark}

\def\be{\begin{equation}}
\def\ee{\end{equation}}
\def\bc{\begin{center}}
\def\ec{\end{center}}

\begin{document}

%----------------------------------------------------------------------------------------
%
%		TITLE & AUTHORS
%
%----------------------------------------------------------------------------------------

\title{Mean field bipartite spin models treated with mechanical techniques}


\author{
Adriano Barra\footnote{Sapienza Universit\`a di Roma, Dipartimento di Fisica and GNFM Gruppo di Roma,  Italy},
Andrea Galluzzi \footnote{Sapienza Universit\`a di Roma, Dipartimento di Matematica and GNFM Gruppo di Roma, Italy},
Francesco Guerra\footnote{Sapienza Universit\`a di Roma, Dipartimento di Fisica and INFN Sezione di Roma,  Italy},
Andrea Pizzoferrato\footnote{The University of Warwick, Mathematics Institute, Coventry, United Kingdom},
Daniele Tantari \footnote{Sapienza Universit\`a di Roma, Dipartimento di Matematica and GNFM Gruppo di Roma, Italy}
}

\maketitle
%----------------------------------------------------------------------------------------
%
%		ABSTRACT
%
%----------------------------------------------------------------------------------------
\begin{abstract}
Inspired by a continuously increasing interest in modeling and framing complex systems in a thermodynamic rationale, in this paper we continue our investigation in adapting well known techniques (originally stemmed in fields of physics and mathematics far from the present) for solving for the free energy of mean field spin models in a statistical mechanics scenario.

Focusing on the test cases of bipartite spin systems embedded with all the possible interactions (self and reciprocal), we show that both the fully interacting bipartite ferromagnet as well as the spin glass counterpart, at least at the replica symmetric level, can be solved via the fundamental theorem of calculus, trough an analogy with the Hamilton-Jacobi theory and lastly with a mapping to a Fourier diffusion problem. All these technologies are shown symmetrically for ferromagnets and spin-glasses in full details and contribute as powerful tools in the investigation of complex systems.
\end{abstract}

%----------------------------------------------------------------------------------------
%
%		INTRODUCTION
%
%----------------------------------------------------------------------------------------
\section{Introduction}
In the last years, equilibrium statistical mechanics has been successfully extended beyond the conventional area of the physics of matter, for instance in quantitative sociology (see e.g. \cite{BC1,galam,HowSMSocial,ExpGSN}) or theoretical biology (see e.g. \cite{GiorgioImmune,amit,bialek,SMautopoietic}). However, these (as well as many others, see e.g. \cite{bouchaud,ghirlanda}) new research fields continuously require more refined mathematical methods and  models in order to give an always more relevant quantitative description and understanding of the phenomena they aim to tackle.
\newline
Among the several novelty these fields of research required, there has been a microscopic description of dynamical systems where two species compete or collaborate, for instance a' la Lotka-Volterra: restricting to equilibrium properties, this need led to a renewal formulation of bipartite spin systems \cite{pierluigi1}, beyond their original introduction within the more standard world of physics of matter \cite{cohen}, which allows to study the emergent collective properties of two interactive large groups of variables. For instance, in quantitative sociology the latter may capture essential features of migrant's integration inside a host community \cite{cecilia} or  the dialogue between two different ensembles of closely interacting cells, as for instance B and T cells within the immune system \cite{galluzzi}.
\newline
In this paper we do not deal with comparing modeling to real data, instead we continue our investigation consisting in obtaining new statistical mechanics techniques all based on adapting existing technologies originally developed to work in field far away from the actual focus, such to make them able to solve for the free energy of suitably defined mean field spin Hamiltonians. We will mainly focus on sum rules originated from a mapping of the statistical mechanics problem with the fundamental theorem of calculus as firstly shown in \cite{hightempGT} and then extended in \cite{interpolBarra}, with the Hamilton-Jacobi framework, firstly developed in \cite{Gsumrules} and then extended in \cite{interpolBarra}\cite{aldo}, and with the Fourier conduction investigated in \cite{MechApproach}\cite{Fourier}, which is a side effect of the mechanical analogy previously introduced.
\newline
Dealing with the subjects and not only with the methodologies, the present work constitutes an extension mainly of \cite{bipCWSK}, where bipartite mean field model have been carefully inspected from a (standard) statistical mechanics perspective and \cite{Fourier} where the techniques we are going to use have been tested on single-party models: the two routes of investigation are here merged together in a unified and stronger theory.
\newline
The paper is divided into two symmetric parts: In the first one, a bipartite ferromagnetic model, which not only considers the interaction among spins of different parties but also between the ones of the same group, is studied through three different interpolation approaches, respectively the fundamental theorem of calculus, the Hamilton-Jacobi scheme and the Fourier transform. In the second one, the same procedures are applied to the disordered (glassy) counterpart of the first model. Unfortunately, as a Parisi-like theory \cite{MPV} for these models is still under construction, and also because it is usually sacrificed in many practical applications involving models beyond the Sherrington-Kirkpatrick paradigm, the thermodynamics of these systems is studied at the replica symmetric level.
%----------------------------------------------------------------------------------------
%
%		FERROMAGNETIC CASE
%
%----------------------------------------------------------------------------------------
\section{Ferromagnetic case}
%----------------------------------------------------------------------------------------
%		THE MODEL - F
%----------------------------------------------------------------------------------------
\subsection{The Model}
The spin system we study is an extension of the one analyzed in \cite{bipCWSK}. There, two dichotomic parties of variables, $\{\sigma_i\}_{i=1,...,N_{\sigma}}$ and $\{\tau_i\}_{i=1,...,N_{\tau}}$, which were coupled through a ferromagnetic interaction, were considered: here we take into account also the ferromagnetic interaction between spins of the same group. All this results in a Hamiltonian made up of the following contributions
\begin{equation}\label{eq:Hfer}
H_{N}\left(\boldsymbol{\sigma,\tau},\boldsymbol{\beta}\right)=-\frac{1}{N}\beta_{\sigma\tau}\sum_{i=1}^{N_{\sigma}}\sum_{j=1}^{N_{\tau}}\sigma_{i}\tau_{j}-\frac{1}{2N}\beta_{\sigma}\sum_{i,j}^{N_{\sigma}}\sigma_{i}\sigma_{j}-\frac{1}{2N}\beta_{\tau}\sum_{i,j}^{N_{\tau}}\tau_{i}\tau_{j},
\end{equation}
where $\sigma_i$, $\tau_i\in\{-1;1\}$  are the two families of dichotomic spin variables;  $\beta_{\sigma}$, $\beta_{\tau}$  and $\beta_{\sigma\tau}$ are the strength of the interactions weighting the intensity of the three different contributions to the Hamiltonian; $N_{\sigma}$ and $N_{\tau}$ are the number of spins for each party with $N=N_{\sigma}+N_{\tau}$. Note that equation ($\ref{eq:Hfer}$) defines a mean-field model, where each couple of spins interact in a ferromagnetic way (all the couplings are positive), and the normalization $1/N$ ensure the linear extensivity of the thermodynamical observables (e.g. energy, entropy, etc.) with  respect to the size of the system.
Introducing  $\alpha= N_{\sigma}/N$ and thus $\left(1-\alpha\right)= N_{\tau}/N$ and denoting with $\mathcal{O}(\boldsymbol{\sigma,\tau})$ a generic observable of the system, the definitions of the statistical mechanic and thermodynamic quantities are given straightforwardly:
\begin{equation*}
\begin{array}{cc}
\text{Partition function} & Z_{N}\left(\boldsymbol{\beta},\alpha\right)\coloneqq\sum_{\boldsymbol{\sigma,\tau}}e^{-\beta H_{N}(\boldsymbol{\sigma,\tau,\beta},\alpha)},\\
&\\
\text{Boltzmann average} & \left<\mathcal{O}\left(\boldsymbol{\sigma,\tau}\right)\right>\coloneqq Z_{N}^{-1}\left(\boldsymbol{\beta},\alpha\right)\sum_{\boldsymbol{\sigma,\tau}}\mathcal{O}\left(\boldsymbol{\sigma,\tau}\right)e^{-\beta H_{N}\left(\boldsymbol{\sigma,\tau,\beta},\alpha\right)} ,\\
&\\
\text{Magnetization of the $\sigma$ party} & m_{\sigma}(\boldsymbol{\sigma})\coloneqq\frac{1}{N_{\sigma}}\sum_{i=1}^{N_{\sigma}}\sigma_{i}, \\
&\\
\text{Magnetization of the $\tau$ party} & m_{\tau}(\boldsymbol{\tau})\coloneqq\frac{1}{N_{\tau}}\sum_{i=1}^{N_{\tau}}\tau_{i}, \\
&\\
\text{Pressure (free energy)} & A(\boldsymbol{\beta},\alpha)=\lim_{N\to\infty}A_{N}\left(\boldsymbol{\beta},\alpha\right)\coloneqq\frac{1}{N}\ln{Z_{N}\left(\boldsymbol{\beta},\alpha\right)}=-\beta f_N(\beta),\\
&\\
\end{array}
\end{equation*}
where $f_N(\beta)$ is the free energy.
\newline
In the following, for the sake of simplicity and without loss of generality, we  will put $\beta =1$: we can restore the dependence by $\beta$ simply rescaling the couplings $\beta_{x}\to\beta\beta_x$, with $x=\sigma,\tau,\sigma\tau$. In the present paper we want to describe three different techniques that can be used to solve the model and in particular to compute the thermodynamic limit of the intensive pressure as  to characterize the thermodynamic states, i.e. the averages (and in general the moments) of the order parameters.
Each one of the three routes approaches the problem from a different perspective but all of these can be thought as proofs of the following
\begin{theorem}\label{thFer}
The thermodynamic limit of the intensive pressure of the full interacting  ferromagnetic bipartite model defined in $(\ref{eq:Hfer})$ reads as
\begin{eqnarray}\label{eq:FerPress}
A(\boldsymbol{\beta},\alpha)&=&\ln2+\alpha\ln\cosh\left(\beta_{\sigma}\alpha\bar{m}_{\sigma}+\beta_{\sigma\tau}\left(1-\alpha\right)\bar{m}_{\tau}\right) +\left(1-\alpha\right)\ln\cosh\left(\beta_{\sigma\tau}\alpha\bar{m}_{\sigma}+\beta_{\tau}\left(1-\alpha\right)\bar{m}_{\tau}\right)+\\
\nonumber&-&\left[\beta_{\sigma\tau}\alpha\left(1-\alpha\right)\bar{m}_{\sigma}\bar{m}_{\tau}+\frac 1 2 \beta_{\sigma}\alpha^{2}\bar{m}_{\sigma}^{2}+\frac 1 2 \beta_{\tau}\left(1-\alpha\right)^{2}\bar{m}_{\tau}^{2}\right],
\end{eqnarray}
where the two quantities $\bar{m}_{\sigma}$ and $\bar{m}_{\tau}$ are the solution of the following system of self-consistent equations
\begin{equation}\label{eq:selfcontT0}
\begin{cases}
\bar{m}_{\sigma}=\tanh\left(\beta_{\sigma}\alpha\bar{m}_{\sigma}+\beta_{\sigma\tau}\left(1-\alpha\right)\bar{m}_{\tau}\right),\\
\bar{m}_{\tau}=\tanh\left(\beta_{\sigma\tau}\alpha\bar{m}_{\sigma}+\beta_{\tau}\left(1-\alpha\right)\bar{m}_{\tau}\right).
\end{cases}
\end{equation}
\end{theorem}
\begin{remark}
Equations $(\ref{eq:selfcontT0})$ can be obtained by extremizing the free energy expressed in Theorem $1$ with respect to the trial parameters $\bar{m}_{\sigma}$ and $\bar{m}_{\tau}$. We stress that where  $\beta_{\sigma}\beta_{\tau}\geq\beta_{\sigma\tau}^{2}$ the optimal parameters impose a maximum for the pressure landscape while in the opposite region it is a saddle point only. On the critical surface $\beta_{\sigma}\beta_{\tau}=\beta_{\sigma\tau}^{2}$ the pressure has a flat direction and the model can be described through a single order parameter that is a linear combination of the two magnetizations, i.e. $\bar{\epsilon}=\sqrt{\beta_{\sigma}}\alpha\bar{m}_{\sigma}+\sqrt{\beta_{\tau}}\left(1-\alpha\right)\bar{m}_{\tau}$ and
\be
A(\boldsymbol{\beta},\alpha)=\ln 2 +\alpha\ln\cosh(\sqrt{\beta_{\sigma}}\bar{\epsilon})+(1-\alpha)\ln\cosh(\sqrt{\beta_{\tau}}\bar{\epsilon})-\frac {\bar{\epsilon}^2} 2
\ee
with $\bar{\epsilon}$ satisfying
\be\label{eqgrondaia}
\bar{\epsilon}=\sqrt{\beta_{\sigma}}\alpha\tanh(\sqrt{\beta_{\sigma}}\bar{\epsilon})+\sqrt{\beta_{\tau}}(1-\alpha)\tanh(\sqrt{\beta_{\tau}}\bar{\epsilon}).
\ee
\end{remark}

%\begin{equation*}
%\begin{cases}
%\frac{\partial A}{\partial\bar{m}_{\sigma}}=\left(\beta_{\sigma}\alpha^{2}\right)\left(\left\langle \sigma_{i}\right\rangle _{t=0}-\bar{m}_{\sigma}\right)+\left(\beta_{\sigma\tau}\alpha\left(1-\alpha\right)\right)\left(\left\langle \tau_{i}\right\rangle _{t=0}-\bar{m}_{\tau}\right)=0\\
%\frac{\partial A}{\partial\bar{m}_{\tau}}=\left(\beta_{\sigma\tau}\alpha\left(1-\alpha\right)\right)\left(\left\langle \sigma_{i}\right\rangle _{t=0}-\bar{m}_{\sigma}\right)+\left(\beta_{\tau}\left(1-\alpha\right)^{2}\right)\left(\left\langle \tau_{i}\right\rangle _{t=0}-\bar{m}_{\tau}\right)=0
%\end{cases}
%\end{equation*}

It is worth noticing that in the limit of $\beta_{\sigma\tau}=0$ the two parties are non interacting and a convex linear combination of two standard Curie-Weiss pressure at suitable temperatures is obtained, while, for $\beta_{\sigma}=\beta_{\tau}=0$, the results developed in \cite{bipCWSK} for a bipartite system without monopartite interactions are recovered.\\
%----------------------------------------------------------------------------------------
%		FIRST APPROACH: FUNDAMENTAL THEOREM OF CALCULUS - F
%----------------------------------------------------------------------------------------
\subsection{First approach: Sum rule}
% Explanation of the technique
The method that in this section we adapt to fully interacting bipartite ferromagnets has been successfully applied in \cite{trial1, rsneuralanalog} for a huge class of single disordered system or systems in reciprocal interactions but without self-contributions. Here we show how it works in the larger case of complete topological interactions, starting with simpler case of the ferromagnetic couplings highlighting the perspective we want to follow. In the second half of the manuscript we will apply it to the disordered counterpart, which will require some more mathematical efforts.
\begin{figure}[!ht] \label{fig:negcoop}
\begin{center}
\includegraphics[width=10cm]{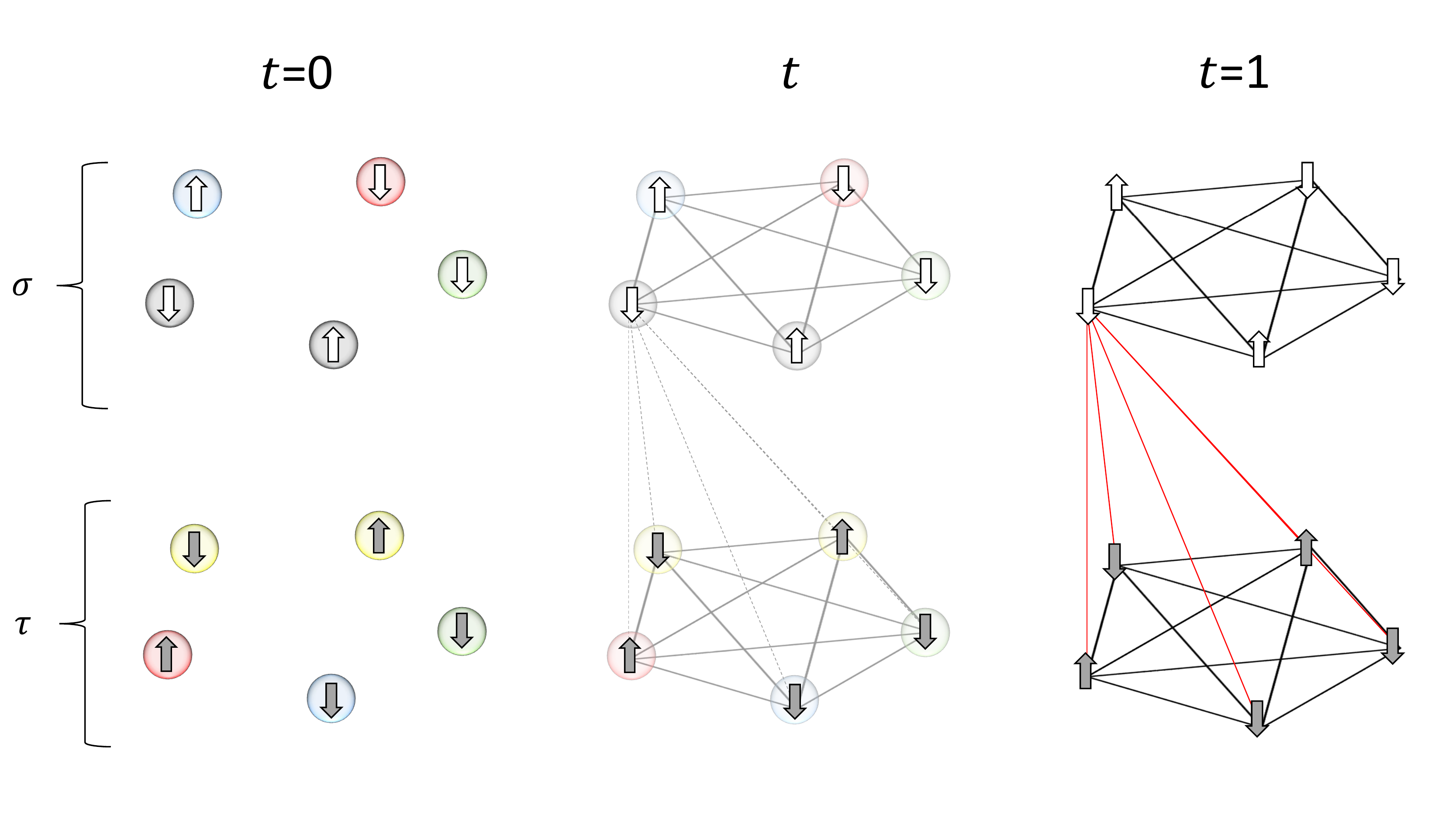}
\end{center}
\caption{Schematic representation of the morphism we perform through interpolation in the sum rule (first technique). The real system is the one on the right, which is obtained whenever $t=1$ is set, while on the left the system at $t=0$ is shown. Note that at $t=0$ sites are no longer communicating, and their reciprocal interactions are replaced by effective local fields, which are represented as colored surrounding spheres (different colors represent different fields). In the middle an intermediate situation with a generic $t$ is shown for the sake of completeness.}
\label{fig:BOLOGNESI}
\end{figure}
In a quick introductional summary, the technique consists of three steps:
\begin{itemize}
\item Through the introduction of an interpolating parameter $t\in [0,1]$, a new trial Hamiltonian is defined as the sum of two pieces, to which it reduces in the limit $t\to1$ and $t \to 0$: the former is the original model, which has to be solved, and the latter is spin system with a simpler one-body interaction with an external effective field that mirrors the real microscopic interactions in a pure mean field fashion, hence
\begin{equation*}
H(t)=\left(t\right) H_{\text{original}} + \left(1-t \right) H_{\text{one-body}}.
\end{equation*}
From the interpolating Hamiltonian, the definitions of interpolating partition function $Z_N(t)$ and pressure $A_N(t)$ naturally follow simply shifting $\exp(-\beta H) \to \exp(-\beta H(t))$.
\item Once the interpolating structure is defined, an interpolating procedure is needed: this role is played by the Fundamental Theorem of Calculus. The key point is that the pressure of the original model can be written as
\begin{equation*}\label{eq:AdA}
A_N=A_N(1)=A_N\left(0\right)+\int_{0}^{1}\frac{\partial A_N}{\partial t}dt.
\end{equation*}
In this way the problem is split into the calculation of two terms: $A_N\left(0\right)$ and $\int_{0}^{1}\frac{\partial A_N(t)}{\partial t}dt$\\
\item $A_N\left(0\right)$ can be easily calculated  because of the factorizability property of a one-body interaction. For what concerns $\frac{\partial A_N(t)}{\partial t}$, it can be written as the sum of a term $\bar{A}$ independent by $t$ and a rest $R(t)$ proportional to the fluctuations of an appropriately chosen order parameter, in such a way that
\begin{equation*}
A_N=\left( A_N(0) + \bar{A}\right)+\int_{0}^{1} R(t) dt
\end{equation*}
where $R(t)$ is the rest including all the fluctuations which one would like to delete or to reduce as much as possible, using the  self-averaging property of the order parameters, when it occurs.
\end{itemize}
In the concrete case of the ferromagnetic model introduced in the previous section (eq.$(1)$), we define the interpolating Hamiltonian  as
%\begin{alignat*}{2}
%H_{N}(t)&= -t\left[\frac{1}{N}\beta_{\sigma\tau}\sum_{i=1}^{N_{\sigma}}\sum_{j=1}^{N_{\tau}}\sigma_{i}\tau_{j}+\frac{1}{2N}\beta_{\sigma}\sum_{i,j}^{N_{\sigma}}\sigma_{i}\sigma_{j}+\frac{1}{2N}\beta_{\tau}\sum_{i,j}^{N_{\tau}}\tau_{i}\tau_{j}\right]+\\
%&-\left(1-t\right)\left[C_{\sigma}\sum_{i=1}^{N_{\sigma}}\sigma_{i}+C_{\tau}\sum_{i=1}^{N_{\tau}}\tau_{i}\right]
%\end{alignat*}
\begin{equation}
H_{N}(t)= -t\left[\frac{1}{N}\beta_{\sigma\tau}\sum_{i=1}^{N_{\sigma}}\sum_{j=1}^{N_{\tau}}\sigma_{i}\tau_{j}+\frac{1}{2N}
\beta_{\sigma}\sum_{i,j}^{N_{\sigma}}\sigma_{i}\sigma_{j}+\frac{1}{2N}\beta_{\tau}\sum_{i,j}^{N_{\tau}}\tau_{i}\tau_{j}\right]
-\left(1-t\right)\left[C_{\sigma}\sum_{i=1}^{N_{\sigma}}\sigma_{i}+C_{\tau}\sum_{i=1}^{N_{\tau}}\tau_{i}\right],
\end{equation}
where $C_{\sigma}$ and $C_{\tau}$ are constants that have to be determined \textit{a posteriori}. At $t=0$ the intensive pressure can be easily be computed as
\begin{eqnarray}\label{feric}
A_{N}\left(0\right)&=&\frac{1}{N}\ln Z_{N}(0)= \frac{1}{N}\ln\left[\sum_{\boldsymbol{\sigma,\tau}}e^{-H_{N}\left(0\right)}\right]\\
&=&\frac{1}{N}\ln\left[\left(\prod_{i=1}^{N_{\sigma}}\sum_{\sigma_i}e^{C_{\sigma}\sigma_{i}}\right)\left(\prod_{i=1}^{N_{\tau}}\sum_{\tau_i}e^{C_{\tau}\tau_{i}}\right)\right]= \ln2+\alpha\ln\cosh\left(C_{\sigma}\right)+\left(1-\alpha\right)\ln\cosh\left(C_{\tau}\right)\nonumber
\end{eqnarray}
%
%\begin{alignat}{2}\label{feric}
%A_{N}\left(0\right)=&\frac{1}{N}\ln Z_{N}(0)=\\
%=&\frac{1}{N}\ln\left[\sum_{\boldsymbol{\sigma,\tau}}e^{-H_{N}\left(0\right)}\right]=\\
%=&\frac{1}{N}\ln\left[\left(\prod_{i=1}^{N_{\sigma}}\sum_{\sigma_i}e^{C_{\sigma}\sigma_{i}}\right)\left(\prod_{i=1}^{N_{\tau}}\sum_{\tau_i}e^{C_{\tau}\tau_{i}}\right)\right]= \\
%=& \ln2+\alpha\ln\cosh\left(C_{\sigma}\right)+\left(1-\alpha\right)\ln\cosh\left(C_{\tau}\right)
%\end{alignat}
Then, the derivative of the pressure with respect to the interpolating parameter is performed as
\begin{eqnarray}\label{eq:spec}
\frac{\partial A_{N}\left(t\right)}{\partial t}&=&\frac{1}{N}\frac{\partial}{\partial t}\left[\ln Z_{N}\left(t\right)\right]=\frac{1}{N}\frac{1}{Z_{N}\left(t\right)}\sum_{\boldsymbol{\sigma,\tau}}\frac{\partial}{\partial t}e^{-H_{N}\left(t\right)}\nonumber \\
&=&\beta_{\sigma\tau}\alpha\left(1-\alpha\right)\left\langle m_{\sigma}m_{\tau}\right\rangle _{t}+\frac{\beta_{\sigma}\alpha^2}{2}\left\langle m_{\sigma}^{2}\right\rangle _{t}+\frac{\beta_{\tau}\left(1-\alpha\right)^2}{2}\left\langle m_{\tau}^{2}\right\rangle _{t} - C_{\sigma}\alpha\left\langle m_{\sigma}\right\rangle _{t}-C_{\tau}\left(1-\alpha\right)\left\langle m_{\tau}\right\rangle _{t}
\end{eqnarray}
Now the last expression has to be written in terms of the fluctuations of the order parameters. Defining $a$, $b$ and $c$ as free coefficients, the generic form of the fluctuations of the order parameters is
\begin{eqnarray}\label{eq:gen}
&&a\left<\left(m_{\sigma}-\bar{m}_{\sigma}\right)\left(m_{\tau}-\bar{m}_{\tau}\right)\right>_{t}+b\left<\left(m_{\sigma}-\bar{m}_{\sigma}\right)^{2}\right>_{t}+c\left<\left(m_{\tau}-\bar{m}_{\tau}\right)^{2}\right>_{t}=\\
&&=a\left<m_{\sigma}m_{\tau}\right>_{t}+b\left<m_{\sigma}^{2}\right>_{t}+c\left<m_{\tau}^{2}\right>_{t}+ \left(-a\bar{m}_{\tau}-2b\bar{m}_{\sigma}\right)\left<m_{\sigma}\right>_{t}+\left(-a\bar{m}_{\sigma}-2c\bar{m}_{\tau}
\right)\left<m_{\tau}\right>_{t}+\left[a\bar{m}_{\sigma}\bar{m}_{\tau}+b\bar{m}_{\sigma}^{2}+c\bar{m}_{\tau}^{2}\right].\nonumber
\end{eqnarray}
Hence, we can identify each coefficient of the equation ($\ref{eq:gen}$) with the ones of the specific expression ($\ref{eq:spec}$), in such a way that we can fix the coefficients $C_{\sigma}$ and $C_{\tau}$ as
\begin{equation*}
\begin{array}{ccc}
C_{\sigma}=\alpha\beta_{\sigma}\bar{m}_{\sigma}+\beta_{\sigma\tau}\left(1-\alpha\right)\bar{m}_{\tau} & ; & C_{\tau}=\beta_{\sigma\tau}\alpha\bar{m}_{\sigma}+(1-\alpha)\beta_{\tau}\bar{m}_{\tau}.
\end{array}
\end{equation*}
Using equations ($\ref{feric}$) and ($\ref{eq:spec}$) we can then write down the following sum rule
\begin{eqnarray}\label{fersumrule}
A_N(\boldsymbol{\beta},\alpha)&=&\ln2+ \alpha\ln\cosh\left(\beta_{\sigma}\alpha\bar{m}_{\sigma}+\beta_{\sigma\tau}\left(1-\alpha\right)\bar{m}_{\tau}\right)
+\left(1-\alpha\right)\ln\cosh\left(\beta_{\sigma\tau}\alpha\bar{m}_{\sigma}+\beta_{\tau}\left(1-\alpha\right)\bar{m}_{\tau}\right)\\
\nonumber&-&\left[\beta_{\sigma\tau}\alpha\left(1-\alpha\right)\bar{m}_{\sigma}\bar{m}_{\tau}+\frac 1 2 \beta_{\sigma}\alpha^{2}\bar{m}_{\sigma}^{2}+\frac 1 2 \beta_{\tau}\left(1-\alpha\right)^{2}\bar{m}_{\tau}^{2}\right]+ R_N(t)
\end{eqnarray}
where
\begin{equation}
R_N(t)=\int_0^1 dt\left[ a\left<\left(m_{\sigma}-\bar{m}_{\sigma}\right)\left(m_{\tau}-\bar{m}_{\tau}\right)\right>_{t}
+b\left<\left(m_{\sigma}-\bar{m}_{\sigma}\right)^{2}\right>_{t}+c\left<\left(m_{\tau}-\bar{m}_{\tau}\right)^{2}\right>_{t}\right].
\end{equation}
Since in the ferromagnetic models the magnetizations are self-averaging in the thermodynamic limit, we can argue that, for a particular choice of the parameters $\bar{m}_{\sigma}$ and $\bar{m}_{\tau}$ (that is by extremizing the pressure with respect to them) we can neglect the rest in $(\ref{fersumrule})$, namely
$$
\lim_{N \to \infty}R_N(t)=0, 
$$
and in the same limit $A_N \to A$ (where $A$ represents the pressure evaluated for $N \to \infty$), that completes the proof of Theorem $\ref{thFer}$.  Note that, by deriving equation ($\ref{fersumrule}$) with respect to $\bar{m}_{\sigma}$ and $\bar{m}_{\tau}$ we get
\begin{equation}\label{ferextr}
\begin{cases}
\frac{\partial A}{\partial\bar{m}_{\sigma}}=\beta_{\sigma}\alpha^{2}\left(\left\langle \sigma_{i}\right\rangle _{t=0}-\bar{m}_{\sigma}\right)+\beta_{\sigma\tau}\alpha\left(1-\alpha\right)\left(\left\langle \tau_{i}\right\rangle _{t=0}-\bar{m}_{\tau}\right)=0,\nonumber\\
\frac{\partial A}{\partial\bar{m}_{\tau}}=\beta_{\sigma\tau}\alpha\left(1-\alpha\right)\left(\left\langle \sigma_{i}\right\rangle _{t=0}-\bar{m}_{\sigma}\right)+\beta_{\tau}\left(1-\alpha\right)^{2}\left(\left\langle \tau_{i}\right\rangle _{t=0}-\bar{m}_{\tau}\right)=0,
\end{cases}
\end{equation}
from which we can argue that, as soon as $\beta_{\sigma}\beta_{\tau}\neq \beta_{\sigma\tau}^2$, the optimal order parameters satisfy
\begin{eqnarray}
\bar{m}_{\sigma}&=& \left\langle \sigma_{i}\right\rangle _{t=0} \nonumber\\
\bar{m}_{\tau}&=& \left\langle \tau_{i}\right\rangle _{t=0},
\end{eqnarray}
i.e. the magnetizations of the interpolating system at $t=0$ are the same of the original system's ones. From the equations ($\ref{ferextr}$) we can see that, on the critical surface $\beta_{\sigma}\beta_{\tau}= \beta_{\sigma\tau}^2$, we have just one single degenerate self consistent equation, that is eq.($\ref{eqgrondaia}$) for an order parameter $\epsilon(\boldsymbol{\sigma,\tau})=\sqrt{\beta_{\sigma}}\alpha m_{\sigma}(\boldsymbol{\sigma,\tau})+\sqrt{\beta_{\tau}}\left(1-\alpha\right)m_{\tau}(\boldsymbol{\sigma,\tau})$ which is a linear combination of the two magnetizations. In this region of the phase space $\epsilon(\boldsymbol{\sigma,\tau})$ is self averaging but the two magnetizations can fluctuate. This phenomenon is very clear for example in the special case in which $\beta_{\sigma}=\beta_{\tau}=\beta_{\sigma\tau}=\beta$, where we cannot distinguish any longer between the two parties: the system is a single Curie Wiess model, of size $N$, characterized by a single order parameter that is the global magnetization $M(\boldsymbol{\sigma,\tau})=\alpha m_{\sigma}(\boldsymbol{\sigma})+(1-\alpha)m_{\tau}(\boldsymbol{\tau})=\beta^{-1/2}\epsilon(\boldsymbol{\sigma,\tau})$.

%----------------------------------------------------------------------------------------
%		SECOND APPROACH: MECHANICAL ANALOGY - F
%----------------------------------------------------------------------------------------
\subsection{Second approach: The Hamilton-Jacobi framework}
% Explanation of the technique
Besides the fundamental theorem of calculus, another interpolation method, developed in \cite{Gsumrules}, can be used. The latter is based on a mechanistic interpretation of the statistical mechanic and thermodynamic quantities defined at the beginning of this section.
\newline
The main idea is the following: the problem of obtaining an explicit expression for the pressure of the model ($\ref{eq:Hfer}$) in the thermodynamic limit and in terms of its order and tunable parameters, is translated in solving an Hamilton-Jacobi equation, where the pressure plays as the action, with suitable boundary conditions.
For this purpose, with the freedom of thinking at the interpolating parameters $t \in \mathcal{R}^+$ and $x \in \mathcal{R}$ as fictitious time and space respectively, we first define an interpolating Hamiltonian as
\begin{eqnarray}
H_{N}\left(t,x\right)=&-&t\left[N_{\tau}\alpha\beta_{\sigma\tau}m_{\sigma}m_{\tau}+\frac{N_{\sigma}}{2}\alpha\beta_{\sigma}m_{\sigma}^{2}+\frac{N_{\tau}}{2}\left(1-\alpha\right)\beta_{\tau}m_{\tau}^{2}\right]\\
&-&\left(1-t\right)\left[\frac{N_{\sigma}}{2}\beta'_{\sigma}m_{\sigma}^{2}+\frac{N_{\tau}}{2}\beta'_{\tau}m_{\tau}^{2}\right] -x N D_N(\sigma,\tau),
\end{eqnarray}
where, introducing a free parameter $a\in(0,\infty)$ (whose practical convenience will be evident later), we  defined
\begin{equation*}
\begin{array}{ccc}
\beta'_{\sigma}=\alpha\left[a^2\beta_{\sigma\tau}+\beta_{\sigma}\right] & ; & \beta'_{\tau}=\left(1-\alpha\right)\left[a^{-2}\beta_{\sigma\tau}+\beta_{\tau}\right],
\end{array}
\end{equation*}
and the order parameter $D(\sigma,\tau)$ (that is just a linear combination of the magnetizations)
\begin{equation}\label{eq:HJD}
D(\sigma,\tau)=\sqrt{\beta_{\sigma\tau}}\left[\alpha a m_{\sigma}-\left(1-\alpha\right)a^{-1}m_{\tau}\right].
\end{equation}
Then, from the definition of the interpolating Hamiltonian we introduce, as usual, an interpolating pressure $S_N(t,x)= N^{-1} \ln \sum_{\sigma \ \tau}e^{-H_N(t,x)}$, which we named $S$ as it plays the role of an action  in the $(t,x)$ space.
\newline
Performing the temporal and spatial derivatives of $S_{N}\left(t,x\right)$ and denoting with a subscript $(t,x)$ the averages performed within the extended Boltzmann measure weighted by $H_N(t,x)$\footnote{Note that this extended average reduces to the canonical one whenever measured at $t=1$ and $x=0$.}, we get
\begin{alignat*}{2}
\frac{\partial S_{N}\left(t,x\right)}{\partial t}=&-\frac{1}{2}\left\langle D^{2}\right\rangle_{(t,x)}\\
\frac{\partial S_{N}\left(t,x\right)}{\partial x}=&\left\langle D\right\rangle_{(t,x)}\\
\frac{\partial^{2}S_{N}\left(t,x\right)}{\partial x^{2}}=&N\left(\left\langle D^{2}\right\rangle_{(t,x)} -\left\langle D\right\rangle_{(t,x)}^{2}\right),
\end{alignat*}
thus, directly by construction, we can write the following Hamilton-Jacobi equation for $S_N(t,x)$:
\be\label{HJFerr}
\partial_t S_N(t,x)+\frac 1 2 (\partial_x S_N(t,x))^2 + V_N(t,x)=0,
\ee
where the potential $V_N(t,x)$ is defined as
\begin{equation}\label{eq:mollifier}
V_N(t,x)= -\frac 1 2 \left(\left\langle D^{2}\right\rangle_{(t,x)} -\left\langle D\right\rangle_{(t,x)}^{2}\right)= \frac 1 {2N} \frac{\partial^{2}S_{N}\left(t,x\right)}{\partial x^{2}}.
\end{equation}
Because of the self-averaging of the order parameters\footnote{Alternatively, instead of assuming self-averaging for the vector $\langle D \rangle$ in the thermodynamic limit ($\langle D \rangle =\lim_{N \to \infty} \langle D_N \rangle$), it is possible to obtain it simply by noticing the $N^{-1}$ pre-factor at the r.h.s. of equation ($\ref{eq:mollifier}$), multiplying a bounded function.}, the potential becomes negligible when the size of the system grows to infinity, hence $S(t,x)$ satisfies, in the thermodynamic limit, a free Hamilton-Jacobi equation.
\newline
We can easily solve it by noting that the velocity field $D(t,x)=\partial_x S(t,x)=\left\langle D \right\rangle_{(t,x)}$ is constant along the trajectories $x=x_0+D(t,x)t$, such that  $D(t,x)$ can be determined from the relation
\be
D(t,x)=D(0,x_0)=\partial_x S(0,x)|x=x_0(t,x)
\ee
that plays as a self consistent equation for $D$.
\newline
The general expression for the action $S(t,x)$ can be obtained as its value evaluated in a point $S(0,x_0)$ (the Cauchy condition) plus the integral of the Lagrangian $\mathcal{L}(t,x)$ over time. Note that here, as the trajectories of such a fictitious motion are straight lines, or alternatively because the potential is zero,
the Lagrangian reads off simply as $\mathcal{L}(t,x) = \langle D \rangle_{t,x}^2 /2$, hence overall we can write
\be
S\left(t,x\right) = S\left(t=0,x=x_0\right) + \int_0^t \mathcal{L}\left(t',x\right)dt' = S\left(0,x_0\right)+\frac t 2 \langle D^2 \rangle_{(t,x)}.
\ee
Note also that the Cauchy starting point implicitly allows for factorization over the sites $\bold{\sigma}, \bold{\tau}$ as at $t=0$ the (potentially tricky) two-body interactions disappear.
\newline
All the quantities we need can then be derived simply by computing the interpolating pressure at $t=0$,  which is
\begin{equation}
S_{N}\left(0,x\right)=\frac{1}{N}\ln\sum_{\boldsymbol{\sigma,\tau} }e^{\left(\beta'_{\sigma}\frac{N_{\sigma}}{2}m_{\sigma}^{2}+\beta'_{\tau}\frac{N_{\tau}}{2}m_{\tau}^{2}\right)} e^{\left(x\sqrt{\beta_{\sigma\tau}}\right)N\left(\alpha a m_{\sigma}-\left(\left(1-\alpha\right) a^{-1}\right)m_{\tau}\right)}
\end{equation}
that is the pressure of two independent Curie-Wiess model with external fields $h$, i.e.
\begin{equation}\label{eq:St0}
S\left(0,x\right)=\alpha A^{CW}\left(\beta'_{\sigma},x a \sqrt{\beta_{\sigma\tau}}\right)+\left(1-\alpha\right)A^{CW}\left(\beta'_{\tau},-x a^{-1}\sqrt{\beta_{\sigma\tau}}\right),
\end{equation}
where
\be
A^{CW}(\beta,h)=\log 2 + \log \cosh(\beta (m+h) )-\frac {\beta} 2 m^2.
\ee
and $m=m(\beta)$ is the solution of the self-consistent equation $m=\tanh(\beta m)$. Taking the derivative with respect to $x$ we get the initial condition for the velocity field
\be
D(0,x)=\partial_x S(0,x)= \sqrt{\beta_{\sigma\tau}}\left[\alpha a m(\beta '_{\sigma},x a \sqrt{\beta_{\sigma,\tau}})-(1-\alpha)a^{-1}m(\beta '_{\tau},-x a^{-1} \sqrt{\beta_{\sigma\tau}})\right].
\ee
At this point we can explicitely write down the self-consistent equation for the velocity field $D(t,x)$ that has to be the solution of
\be\label{velfield0}
D(t,x)=D(0,x_0)=\sqrt{\beta_{\sigma\tau}}\left[\alpha a m(\beta '_{\sigma},(x-D(t,x)t) a \sqrt{\beta_{\sigma,\tau}})-(1-\alpha)a^{-1}m(\beta '_{\tau},-(x-D(t,x)t) a^{-1} \sqrt{\beta_{\sigma\tau}})\right].
\ee
Finally, remembering that $x_0  = x - D(t,x) t$,  the pressure of the model can be written in terms of $D(t,x)$ as
\be\label{HJpressexp}
S(t,x)=S(0,x-D(x,t)t)+\frac t 2 D^2(t,x).
\ee
It is easy to check that, whenever evaluated at $t=1$ and $x=0$ the expression $(\ref{HJpressexp})$ does coincide with the expression  $(\ref{eq:FerPress})$  obtained through the first method. In fact, referring to equation ($\ref{velfield0}$), we can define
\begin{eqnarray}
\bar{m}_{\sigma}(\boldsymbol{\beta};a)&=&m(\beta '_{\sigma},-D(1,0)a \sqrt{\beta_{\sigma,\tau}}),\nonumber \\
\bar{m}_{\tau} (\boldsymbol{\beta};a)&=&m(\beta '_{\tau},D(1,0) a^{-1} \sqrt{\beta_{\sigma\tau}}).
\end{eqnarray}
Since $D(\boldsymbol{\beta})=D(1,0)=\sqrt{\beta_{\sigma\tau}}\left[\alpha a \bar{m}_{\sigma}-(1-\alpha a^{-1}m_{\tau})\right]$, $\bar{m}_{\sigma}$ and  $\bar{m}_{\tau}$ satisfy the following system of coupled equations
\begin{eqnarray}
\bar{m}_{\sigma}&=&m(\beta '_{\sigma},-D a \sqrt{\beta_{\sigma,\tau}})=\tanh(\alpha \beta_{\sigma}\bar{m}_{\sigma}+ (1-\alpha)\beta_{\sigma\tau}\bar{m}_{\tau} ),\nonumber \\
\bar{m}_{\tau}&=&m(\beta '_{\tau},D a^{-1} \sqrt{\beta_{\sigma,\tau}})=\tanh((1-\alpha) \beta_{\tau}\bar{m}_{\tau}+ \alpha\beta_{\sigma\tau}\bar{m}_{\sigma} ),
\end{eqnarray}
that is exactly the system defining the order parameters in the first method and reported in eq. $(\ref{eq:selfcontT0})$. Using this decomposition for $D(\boldsymbol{\beta})$ we can rewrite equation$(\ref{HJpressexp})$ once again obtaining  for the pressure the expression $(2)$, in full agreement with Theorem $1$ statements.
\newline
Note that, as it should be, since  $\bar{m}_{\sigma}(\boldsymbol{\beta};a)$ and $\bar{m}_{\tau}(\boldsymbol{\beta};a)$ do not depend on $a$, we get the same expression for the pressure of the model independently by the choice of $a$ in the interpolating procedure. This degree of freedom allows us to give a physical interpretation to the quantities $\bar{m}_{\sigma}$ and $\bar{m}_{\tau}$ as magnetizations ''completely inside'' the Hamilton-Jacobi framework. In fact, since
\be
D(\boldsymbol{\beta})=\sqrt{\beta_{\sigma\tau}}\left[\alpha a \left\langle m_{\sigma}\right\rangle -(1-\alpha)a^{-1}\left\langle m_{\tau}\right\rangle\right]=\sqrt{\beta_{\sigma\tau}}\left[\alpha a \bar{m}_{\sigma} -(1-\alpha)a^{-1} \bar{m}_{\tau}\right],
\ee
for every choice of the parameter $a$, we obtain $ \left\langle m_{\sigma}\right\rangle=\bar{m}_{\sigma}$ and $ \left\langle m_{\tau}\right\rangle=\bar{m}_{\tau}$.

\begin{remark}
We can use fruitfully the freedom in the choice of the free parameter $a$ by imposing that the velocity is zero when $x=0$. In this way $S(t,x)=S(0,x_0)$, i.e. the pressure of the model can be written as a convex linear combination of two non interacting single-party systems at suitable temperatures. We can do that by imposing
\be
\sqrt{\beta_{\sigma\tau}}\left[\alpha a \bar{m}_{\sigma} -(1-\alpha)a^{-1} \bar{m}_{\tau}\right]=0,
\ee
i.e. choosing $a=\sqrt{\frac{(1-\alpha)\bar{m}_{\tau}}{\alpha\bar{m}_{\sigma}}}$. In this way we can write
\be
A(\beta_{\sigma},\beta_{\tau},\beta_{\sigma\tau})=\alpha A^{CW}(\beta '_{\sigma})+(1-\alpha)A^{CW}(\beta '_{\tau}).
\ee
This result generalizes the decomposition introduced for the first time in \cite{bipCWSK}, concerning the bipartite systems without self interactions.
\end{remark}

%----------------------------------------------------------------------------------------
%		THIRD APPROACH: FOURIER FRAMEWORK - F
%----------------------------------------------------------------------------------------
\subsection{Third approach: The Fourier framework}
In line with the precedent remark, in this section we show a strategy easily obtainable revisiting the Hamilton-Jacobi scheme. In fact, instead of giving the solution of the model through ($\ref{velfield0}$) and($\ref{HJpressexp}$), the (Cole-Hopf transform of the)  function $S_N\left(t,x\right)$ can be studied in its conjugate Fourier space $\left(t,k\right)$ and solved via standard Green function plus convolution theorem route as summarized in the following adaptation of the Lax Theorem \cite{lax}.
%Lemma
\begin{theorem}\label{le:Fourier}
Using $S_0(x)$ to quantify the value of the action at $t=0$, and a subscript $N$ to denote averages of observable performed at finite $N$ with its lacking accounting for quantities evaluated in the thermodynamic limit, then for $N \to \infty$ the solution of
\begin{equation*}
\begin{cases}
\partial_tS_{N}\left(t,x\right)+\frac{1}{2}\left(\partial_xS_{N}\left(t,x\right)\right)^2+\frac{1}{2N}\partial^2_{x}S_{N}\left(t,x\right)= 0, & \\
S_{N}\left(0,x\right)=S_0(x), &
\end{cases}
\end{equation*}
namely an explicit expression for the action $S(t,x)$, and the associated Burger problem
\begin{equation*}
\begin{cases}
\partial_t D_N\left(t,x\right)+D_N\left(t,x\right)\partial_x D_N\left(t,x\right)+\frac{1}{2N}\partial^2_{x}D_N\left(t,x\right)= 0, & \\
D_N\left(0,x\right)=D_0(x), &
\end{cases}
\end{equation*}
is given by the Legendre transform of its Cauchy condition on the action, hence
\begin{equation}\label{eq:finFourier}
S\left(t,x\right) = \inf_{y}\left\{ \frac{\left(x-y\right)^2}{2t} + S_0\left(y\right)\right\}=\frac{\left(x-\hat{y}\right)^2}{2t} + S_0\left(\hat{y}\right),
\end{equation}
with $\hat{y}$ minimizer and $x=\hat{y}+D_0(\hat{y})t$.
\end{theorem}
\begin{proof}
First we perform the following Cole-Hopf transform on $S_N\left(t,x\right)$
\begin{equation}
\Psi_N(t,x)\coloneqq e^{-N S_N\left(t,x\right)},
\end{equation}
that, by definition, satisfies the following heat equation:
\begin{equation}
\frac{\partial \Psi_N(t,x)}{\partial t} - \frac{1}{2N} \frac{\partial^2 \Psi_N(t,x)}{\partial x^2} =0.
\end{equation}
Now, calling its Fourier transform $\hat{\Psi}_N(t,k)$, clearly
\begin{equation}
\partial_t \hat{\Psi}_N(t,k) + \frac{k^2}{2N}\hat{\Psi}_N(t,k)=0,
\end{equation}
whose solution is given by
\begin{equation}
\hat{\Psi}_N(t,k) = \hat{\Psi}_0(k) \exp( - \frac{k^2}{2 N} t ).
\end{equation}
Coming back to the original space we get
\begin{equation*}\label{psi}
\Psi_N(t,x)=  \int d y\, G_t(x-y)  \Psi_0(y) = \sqrt{\frac{N}{2 \pi t}}  \int d y  \,  e^{-N\frac{(x-y)^2}{2t} } \Psi_0(y)
\end{equation*}
where $G_t(x-y)$ is the Green propagator. Recalling the definition of $\Psi_N(t,x)$ we get
\begin{equation}
S_N\left(t,x\right) = - \frac{1}{N} \log \Psi_N (t,x) = - \frac{1}{N} \ln\sqrt{\frac{N}{2 \pi t}}  \int d y \, e^{-N\left( \frac{(x-y)^2}{2t} + S_{0}\left(y\right)\right)}
\end{equation}
that can be computed in the thermodynamic limit through the saddle-point technique obtaining
\begin{equation}
S(t,x) = \inf_{y}\left\{ \frac{(x-y)^2}{2t} + S_0\left(y\right)\right\}.
\end{equation}
\end{proof}
Using the explicit definition of $S(t,x)$, once computed the initial condition ($\ref{eq:St0}$), we can use Lemma $\ref{le:Fourier}$ and recover exactly equation ($\ref{HJpressexp}$) from which all the considerations of the previous section hold.

%----------------------------------------------------------------------------------------
%
%		DISORDERED CASE: Replic Symmetric Approximation
%
%----------------------------------------------------------------------------------------
\section{Disordered case: Replica Symmetric Approximation}
%----------------------------------------------------------------------------------------
%		THE MODEL - D
%----------------------------------------------------------------------------------------
\subsection{The Model}
In this second part of the paper we study a fully interacting bipartite spin glass. Namely we investigate the disordered counterpart of the model (\ref{eq:Hfer}) where now the coupling may assume  both positive and negative values allowing for frustration. Thus, besides a different normalization of the Hamiltonian in order to ensure the standard extensive linear growth  of the thermodynamical observables with the size of the system, the exchange interactions now are independently drawn at random from a Gaussian distribution $\mathcal{N}(0,1)$, hence
\begin{equation}\label{eq:Hsk}
H_N(\sigma,\tau;\mathbf{J})=-\frac{1}{\sqrt{N}}\beta_{\sigma\tau}\sum_{i=1}^{N_{\sigma}}\sum_{j=1}^{N_{\tau}}J_{ij}^{\sigma\tau}\sigma_{i}\tau_{j}-\frac{1}{\sqrt{2N_{\sigma}}}\beta_{\sigma}\sum_{i,j}^{N_{\sigma}}J_{ij}^{\sigma}\sigma_{i}\sigma_{j}-\frac{1}{\sqrt{2N_{\tau}}}\beta_{\tau}\sum_{i,j}^{N_{\tau}}J_{ij}^{\tau}\tau_{i}\tau_{j},
\end{equation}
The factor $1/\sqrt{2}$, when present, ensures the contribution of each couple of spins to count just once. Further, as in the ferromagnetic case, each contribution is weighted with a $\beta$-parameter, modulating the relative strength between interactions of different nature (mono-partite or bipartite) and within each party. Then, one can define easily the statistical mechanics machinery as before, this time introducing replicas too. Thus, using $\mathbb{E}$ to depict the average over the Gaussian couplings, we have:
\begin{equation*}
\begin{array}{cc}
\text{Partition function} & Z_N=\sum_{\{\sigma,\tau\}}e^{-\beta H_N(\sigma,\tau;\mathbf{J})}\\
&\\
\text{Boltzmann average} & \omega_N(\mathcal{O};\mathbf{J})=Z_N^{-1}\sum_{\{\sigma,\tau\}} \mathcal{O}e^{-\beta H_N(\sigma,\tau;\mathbf{J})}\\
&\\
\text{Product measure over $S$ replicas} & \Omega=\omega_{1} \otimes ...\otimes \omega_{s}\\
&\\
\text{Quenched state} & \left<\mathcal{O} \right>=\mathbb{E}\left[\Omega \left(\mathcal{O} \right)\right]\\
&\\
\text{Overlap of the $\sigma$ party} & q_{\sigma^a\sigma^b}=\frac{1}{N_{\sigma}}\sum_i\sigma_{i}^{a}\sigma_{i}^{b} \\
&\\
\text{Overlap of the $\tau$ party} & q_{\tau^{a}\tau^{b}}=\frac{1}{N_{\tau}}\sum_\mu \tau_{i}^{a}\tau_{i}^{b} \\
&\\
\text{Quenched intensive pressure} & A\left(\alpha,\boldsymbol{\beta}\right)=\lim_{N \to \infty} A_{N}\left(\beta_{\sigma},\beta_{\tau},\beta_{\sigma\tau}\right)=\lim_{N \to \infty} \frac{1}{N}\mathbb{E}\ln Z_{N}\left(\beta_{\sigma},\beta_{\tau},\beta_{\sigma\tau}\right).
\end{array}
\end{equation*}
%
%\begin{equation*}
%\begin{array}{cc}
%\text{Overlap of the $\sigma$ party} & q_{\sigma^a\sigma^b}=\frac{1}{N_{\sigma}}\sum_i\sigma_{i}^{a}\sigma_{i}^{b} \\
%&\\
%\text{Overlap of the $\tau$ party} & q_{\tau^{a}\tau^{b}}=\frac{1}{N_{\tau}}\sum_\mu \tau_{i}^{a}\tau_{i}^{b} \\
%&\\
%\text{Boltzmann measure} & \omega_N(\mathcal{O};\mathbf{J})=Z_N^{-1}\sum_{\{\sigma,\tau\}} \mathcal{O}e^{-\beta H_N(\sigma,\tau;\mathbf{J})}\\
%&\\
%\text{Partition function} & Z_N=\sum_{\{\sigma,\tau\}}e^{-\beta H_N(\sigma,\tau;\mathbf{J})}\\
%&\\
%\text{Product measure over $S$ replicas} & \Omega=\omega_{1} \otimes ...\otimes \omega{s}\\
%&\\
%\text{Quenched state} & \left<\mathcal{O} \right>=\mathbb{E}\left[\Omega \left(\mathcal{O} \right)\right]\\
%&\\
%\text{Quenched intensive pressure} & A\left(\alpha,\boldsymbol{\beta}\right)=\lim_{N \to \infty} A_{N}\left(\beta_{\sigma},\beta_{\tau},\beta_{\sigma\tau}\right)=\lim_{N \to \infty} \frac{1}{N}\mathbb{E}\ln Z_{N}\left(\beta_{\sigma},\beta_{\tau},\beta_{\sigma\tau}\right)
%\end{array}
%\end{equation*}
As usual the (quenched) free energy $f(\alpha,\beta)$ is related to the (quenched) pressure $A(\alpha,\beta)$ via $A(\alpha,\beta)=-\beta f(\alpha,\beta)$.
Note that in the rest of the paper we will set $\beta=1$ without loss of generality as we can restore the dependence  by $\beta$ simply rescaling the couplings $\beta_x\to\beta\beta_x$, with $x=\sigma,\tau,\sigma\tau$.
As in the first part of the paper, the expression of the quenched pressure in the replica symmetric approximation is determined using the three different techniques described before. Nevertheless, the presence of the overlaps instead of the magnetizations of the spins implies slightly different procedures in the proofs with respect to the ferromagnetic case. Even so, all the strategies produce the same solution as stated in the following
\begin{theorem}\label{thSK}
The Replica Symmetric Approximation for the intensive pressure of the model defined in $(\ref{eq:Hsk})$ reads as
\begin{eqnarray}\label{eq:SKPress}
\nonumber && A^{RS}(\alpha,\boldsymbol{\beta})=\ln2 \\\label{Ars}
\nonumber && +\alpha \int d\mu(z)\ln\cosh\left(z\sqrt{\left(\beta_{\sigma}^{2}\right)\bar{q}_{\sigma\sigma'}+\beta_{\sigma\tau}^{2}\left(1-\alpha\right)\bar{q}_{\tau\tau'}}\right)+
\left(1-\alpha\right)\int d\mu (z)\ln\cosh\left(z\sqrt{\left(\beta_{\sigma\tau}^{2}\alpha\right)\bar{q}_{\sigma\sigma'}+\left(\beta_{\tau}^{2}\right)\bar{q}_{\tau\tau'}}\right)+\\
\nonumber && +\frac{\beta_{\sigma}^{2}}{4}\alpha\left(\bar{q}_{\sigma\sigma'}-1\right)^{2}+\frac{\beta_{\tau}^{2}}{4}\left(1-\alpha\right)\left(\bar{q}_{\tau\tau'}-1\right)^{2}+\frac{1}{2}\beta_{\sigma\tau}^{2}\alpha\left(1-\alpha\right)\left(1-\bar{q}_{\sigma\sigma'}\right)\left(1-\bar{q}_{\tau\tau'}\right),
\end{eqnarray}
where $d\mu(z)$ is a unitary gaussian measure and the order parameters $\bar{q}_{\sigma\sigma'}$ and $\bar{q}_{\tau\tau'}$ are the solutions of the following system of self-consistent coupled equations
\begin{equation}
\begin{cases}\label{dissc} \bar{q}_{\sigma\sigma'}=\int d\mu(z)\tanh^2\left(z\sqrt{\beta_{\sigma}^{2}\bar{q}_{\sigma\sigma'}+\beta_{\sigma\tau}^{2}\left(1-\alpha\right)\bar{q}_{\tau\tau'}}\right),\\
\bar{q}_{\tau\tau'}=\int d\mu(z)\tanh^2\left(z\sqrt{\beta_{\sigma\tau}^{2}\alpha\bar{q}_{\sigma\sigma'}+\beta_{\tau}^{2}\bar{q}_{\tau\tau'}}\right).\\
\end{cases}
\end{equation}
Finally, in the region
\be\label{convcond}
\beta_{\sigma}\beta_{\tau}\geq\beta_{\sigma\tau}^{2}\sqrt{\alpha(1-\alpha)},
\ee
the following sum rule holds
\be
A(\alpha,\boldsymbol{\beta})\leq A^{RS}(\alpha,\boldsymbol{\beta}).
\ee
\end{theorem}
As it will be clear in the next sections, the replica symmetric approximation can be defined assuming that the fluctuations of the order parameters of the model, i.e. $\left\langle q^2_{\sigma\sigma'}\right\rangle-\left\langle q_{\sigma\sigma'}\right\rangle^2$ and $\left\langle q^2_{\tau\tau'}\right\rangle-\left\langle q_{\tau\tau'}\right\rangle^2$, can be neglected in the thermodynamic limit (hence the order parameters are self-averaging quantities). This assumption was exact in the ferromagnetic model but of course it is no longer true in the low noise region of the phase diagram for the disordered counterpart \cite{MPV}: Indeed the well known phenomenon of replica symmetry breaking, clearly understood for single species \cite{parisi2}\cite{parisi3}\cite{broken}\cite{t4}, occurs also in multi-specie spin-glasses \cite{MultiSpecies}, but a Parisi-like theory in this case is still missing, hence, we will focus only on replica symmetric regimes, which, for practical purposes, are generally the standard level of description \cite{amit,ton}.
\newline
As in the ferromagnetic case, when $\beta_{\sigma\tau}=0$ the sum of two independent spin-glass replica symmetric solutions (namely of the Sherrington-Kirkpatrick (SK) type \cite{MPV}) is obtained and for $\beta_{\sigma}=\beta_{\tau}=0$ the same representation of the free energy as the one shown in \cite{bipCWSK} for a purely bipartite interaction is founded.\\
As a last remark before proving Theorem $\ref{thSK}$, note that the  condition ($\ref{convcond}$), as shown in \cite{MultiSpecies}, plays a very important role for a lot of issues including the proof of the existence of the thermodynamic limit and the convexity of the variational principle regulating the free energy. Finally, exactly as happened in the previous ferromagnetic counterpart, we will  see that, on the critical surface $\beta_{\sigma}\beta_{\tau}=\beta_{\sigma\tau}^{2}\sqrt{\alpha(1-\alpha)}$, a complete description of the model needs just one single order parameter  that is a linear combination of the two overlaps: This phenomenon can be easily understood when $\beta_{\sigma}/\sqrt{\alpha}=\beta_{\tau}/\sqrt{1-\alpha}=\beta_{\sigma\tau}$, i.e. when all the interactions have the same strength: we can't distinguish the two parties and the system reduces just to a single SK spin glass model that can be described completely through a single order parameter that is the global overlap.

%Thus, as in the ferromagnetic case, the gutter condition:
%Allows the system to be expressed as:
%\begin{alignat*}{2}
%A=&\ln2+\\
%&+\alpha\mathbb{E}\ln\left(\cosh\left(\sqrt{\left\{ %\left(\beta_{\sigma}^{2}\right)\bar{q}_{\sigma\sigma'}+\beta_{\sigma\tau}^{2}\left(1-\alpha\right)\bar{q}_{\tau\tau'}\right\} %}\right)\right)+\\
%&+\left(1-\alpha\right)\mathbb{E}\ln\left(\cosh\left(\sqrt{\frac{\left(\beta_{\sigma\tau}^{2}\alpha\right)}{\beta_{\sigma}^{2}}\left\{ %\left(\beta_{\sigma}^{2}\right)\bar{q}_{\sigma\sigma'}+\beta_{\sigma\tau}^{2}\left(1-\alpha\right)\bar{q}_{\tau\tau'}\right\} %}\right)\right)+\\
%&+\frac{1}{4}\frac{\beta_{\tau}\sqrt{1-\alpha}}{\beta_{\sigma\tau}^{2}\left(1-\alpha\right)}\left\{ %\beta_{\sigma}^{2}\left(\bar{q}_{\sigma\sigma'}-1\right)+\beta_{\sigma\tau}^{2}\left(1-\alpha\right)\left(\bar{q}_{\tau\tau'}-1\right)\right\}% ^{2}
%\end{alignat*}

%----------------------------------------------------------------------------------------
%		FIRST APPROACH: FUNDAMENTAL THEOREM OF CALCULUS - D
%----------------------------------------------------------------------------------------
\subsection{First approach: Sum rule}

In this section we give a first proof of Theorem $(\ref{thSK})$ by interpolating  between the original model with nasty two-body couplings and a system regulated by a suitable one-body Hamiltonian whose spins feel an effective random external field representing -at least on average- their microscopic surrounding. This procedure allows to obtain, via the Fundamental Theorem of Calculus, a sum rule for the free energy where overlap fluctuations are embedded in a source term, split from the rest (which, as a consequence, naturally returns the replica symmetric approximation once neglected).
\newline
First of all we define the interpolating Hamiltonian as
\begin{alignat}{2}
H_{N}\left(t\right)=&-\sqrt{t}\left[\frac{\beta_{\sigma\tau}}{\sqrt{N}}\sum_{i=1}^{N_{\sigma}}\sum_{j=1}^{N_{\tau}}
J_{ij}^{\sigma\tau}\sigma_{i}\tau_{j}+\frac{\beta_{\sigma}}{\sqrt{2N_{\sigma}}}\sum_{i,j}^{N_{\sigma}}J_{ij}^{\sigma}\sigma_{i}\sigma_{j}+\frac{\beta_{\tau}}{\sqrt{2N_{\tau}}}
\sum_{i,j}^{N_{\tau}}J_{ij}^{\tau}\tau_{i}\tau_{j}\right]+\\
&-\sqrt{1-t}\left[\sqrt{C_{\sigma}}\sum_{i}\eta_{i}^{\sigma}\sigma_{i}+\sqrt{C_{\tau}}\sum_{\tau}\eta_{i}^{\tau}\tau_{i}\right],
\end{alignat}
where the $\{\eta_{i}^{\sigma}\}_{i=1,..., N_{\sigma}}$ and the $\{\eta_{i}^{\tau}\}_{i=1,..., N_{\tau}}$ are two families of independent unitary gaussian random variables, independent also from the $\mathbf{J}$ and $C_{\sigma}$ and $C_{\tau}$ are two constants that we have to fix appropriately.
Defining naturally the interpolating partition function $Z_N(t)$ and the quenched pressure $A_N(t)$ as
\begin{equation}
Z_N\left(t\right)=\sum_{\left\{ \sigma\right\} }\sum_{\left\{ \tau\right\} }e^{-H_{N}\left(t\right)} \ \ \ \ \ A_N(t)= \frac 1 N\mathbb{E}\ln Z_N(t),
\end{equation}
we recover the original pressure at $t=1$ while, at $t=0$ we have a simpler one body problem that factorizes over the sites and whose pressure can be easily computed and reads as
\begin{eqnarray}\label{disa0}
A_N\left(0\right)&=&\frac{1}{N}\mathbb{E}\ln Z\left(0\right)= \frac{1}{N}\mathbb{E}\ln\sum_{\left\{ \sigma\right\} }\sum_{\left\{ \tau\right\} }e^{-H\left(0\right)}=\nonumber\\
&=&\ln2+\alpha\int d\mu(z)\ln\cosh\left(z\sqrt{C_{\sigma}}\right)+\left(1-\alpha\right)\int d\mu(z)\ln\cosh\left(z\sqrt{C_{\tau}}\right).
\end{eqnarray}
The calculation leading to an explicit expression of $\partial_t A_N(t)$ is long but straightforward and returns
\begin{alignat}{2}
\frac{\partial A_N\left(t\right)}{\partial t}=&-\left[\frac{1}{2}\beta_{\sigma\tau}^{2}\alpha\left(1-\alpha\right)\left\langle q_{\sigma\sigma'}q_{\tau\tau'}\right\rangle +\frac{\beta_{\sigma}^{2}}{4}\alpha\left\langle q_{\sigma\sigma'}^{2}\right\rangle +\frac{\beta_{\tau}^{2}}{4}\left(1-\alpha\right)\left\langle q_{\tau\tau'}^{2}\right\rangle \right]\nonumber \\
&+\left[\frac{C_{\sigma}}{2}\alpha\left\langle q_{\sigma\sigma'}\right\rangle +\frac{C_{\tau}}{2}\left(1-\alpha\right)\left\langle q_{\tau\tau'}\right\rangle \right]\nonumber \\
&+\left[-\frac{C_{\sigma}}{2}\alpha-\frac{C_{\tau}}{2}\left(1-\alpha\right)+\frac{\beta_{\tau}^{2}}{4}\left(1-\alpha\right)+\frac{\beta_{\sigma}^{2}}{4}\alpha+\frac{1}{2}\beta_{\sigma\tau}^{2}\alpha\left(1-\alpha\right)\right].
\end{alignat}
Hence, if we choose
\begin{equation}
\begin{array}{ccc}
C_{\sigma}=\left(\beta_{\sigma}^{2}\right)\bar{q}_{\sigma\sigma'}+\beta_{\sigma\tau}^{2}\left(1-\alpha\right)\bar{q}_{\tau\tau'} & ; & C_{\tau}=\left(\beta_{\sigma\tau}^{2}\alpha\right)\bar{q}_{\sigma\sigma'}+\left(\beta_{\tau}^{2}\right)\bar{q}_{\tau\tau'},
\end{array}
\end{equation}
we can write down the $t$-streaming as
\begin{eqnarray}\label{disdt}
\frac{\partial A_N\left(t\right)}{\partial t}&=&\alpha\frac{\beta_{\sigma}^2}{4}(1-\bar{q}_{\sigma\sigma '})^2+(1-\alpha)\frac{\beta_{\tau}^2}{4}(1-\bar{q}_{\tau\tau '})^2+\alpha(1-\alpha)\frac{\beta_{\sigma\tau}^2}{2}(1-\bar{q}_{\sigma\sigma '})(1-\bar{q}_{\tau\tau'})\nonumber \\
&-&\left[\alpha\frac{\beta_{\sigma}^2}{4}\left\langle (q_{\sigma\sigma'}-\bar{q}_{\sigma\sigma '})^2\right\rangle_t+(1-\alpha)\frac{\beta_{\tau}^2}{4}\left\langle (q_{\tau\tau'}-\bar{q}_{\tau\tau '})^2\right\rangle_t+\alpha(1-\alpha)\frac{\beta_{\sigma\tau}^2}{2}\left\langle (q_{\sigma\sigma'}-\bar{q}_{\sigma\sigma '})(q_{\tau\tau'}-\bar{q}_{\tau\tau'})\right\rangle_t\right]\nonumber.
\end{eqnarray}
Using equation $(\ref{disa0})$ and the last expression for the $t$-streaming of $A(t)$, we can then build the following sum-rule
\begin{equation}\label{dissumrule}
A_N(\alpha,\boldsymbol{\beta})=A_N(1)=A_N(0)+\int_0^1 dt\frac{d}{dt}A_N(t) = A^{RS}(\bar{q}_{\sigma\sigma '},\bar{q}_{\tau\tau '})-\int_0^1 R_N(t),
\end{equation}
where $A^{RS}(\bar{q}_{\sigma\sigma '},\bar{q}_{\tau\tau '})$ is the function stated in Theorem $3$ for a generic couple of parameters $\bar{q}_{\sigma\sigma '}$ and $\bar{q}_{\tau\tau '}$, while the source of overlap fluctuations reads as the rest
\be
 R_N(t)=\alpha\frac{\beta_{\sigma}^2}{4}\left\langle (q_{\sigma\sigma'}-\bar{q}_{\sigma\sigma '})^2\right\rangle_t+(1-\alpha)\frac{\beta_{\tau}^2}{4}\left\langle (q_{\tau\tau'}-\bar{q}_{\tau\tau '})^2\right\rangle_t+\alpha(1-\alpha)\frac{\beta_{\sigma\tau}^2}{2}\left\langle (q_{\sigma\sigma'}-\bar{q}_{\sigma\sigma '})(q_{\tau\tau'}-\bar{q}_{\tau\tau'})\right\rangle_t.
\ee
As soon as $\beta_{\sigma}\beta_{\tau}\geq\beta_{\sigma\tau}^{2}\sqrt{\alpha(1-\alpha)}$, such a source is positively defined and we can minimize the error we commit keeping only the replica-symmetric approximation simply by finding the values of the order parameters that minimize $A^{RS}(\bar{q}_{\sigma\sigma '},\bar{q}_{\tau\tau '})$. By extremizing $A^{RS}(\bar{q}_{\sigma\sigma '},\bar{q}_{\tau\tau '})$ with respect to $\bar{q}_{\sigma\sigma '}$ and $\bar{q}_{\sigma\sigma '}$ we find the conditions $(\ref{dissc})$ that complete the proof of Theorem $\ref{thSK}$. Note that, in the language of the current interpolating method, the equations $(\ref{dissc})$ for the order parameters can be written in the following form
\begin{eqnarray}
\bar{q}_{\sigma\sigma '}&=&\left\langle q_{\sigma\sigma '}\right\rangle_{t=0}\nonumber\\
\bar{q}_{\tau\tau '}&=&\left\langle q_{\tau\tau '}\right\rangle_{t=0}.
\end{eqnarray}
This means that the optimal order parameters represent the mean of the system's overlap when $t=0$. This shows a sort of stochastic stability \cite{pierluigi} in the interpolating procedure and justifies the definition of $A^{RS}(\alpha,\boldsymbol{\beta})$ also in the region $\beta_{\sigma}\beta_{\tau}\leq\beta_{\sigma\tau}^{2}\sqrt{\alpha(1-\alpha)}$, where we don't know the sign of the error term. Finally we want just to point out that, in this interpolating framework, the name "replica symmetric approximation" is justified by the sum rule $(\ref{dissumrule})$, but $A^{RS}(\alpha,\boldsymbol{\beta})$ is the true pressure of the model only if the error term vanishes in the thermodynamic limit, i.e. only in the region of the phase space where the overlaps are self-averaging (high temperature limit \cite{MPV}).

%----------------------------------------------------------------------------------------
%		SECOND APPROACH: MECHANICAL ANALOGY - D
%----------------------------------------------------------------------------------------
\subsection{Second approach: The Hamilton-Jacobi framework}

In this section, as in the ferromagnetic case, we give a proof of Theorem $\ref{thSK}$ using a mechanical analogy with an Hamilton-Jacobi problem for a free particle \footnote{Restricting to single-specie spin glasses, the phenomenon of replica symmetry breaking within the Hamilton-Jacobi framework has been solved and has been reported in \cite{aldo}. For multi-species spin-glasses a quantitative description of such a phenomenon is still lacking. A first trial can be found in \cite{MultiSpecies}.}. First of all we define a  (fictitious) time and space dependent Hamiltonian
\begin{alignat*}{2}
H_{N}\left(t,x\right)=&-\sqrt{t}\left[\frac{1}{\sqrt{N}}\beta_{\sigma\tau}\sum_{i=1}^{N_{\sigma}}\sum_{j=1}^{N_{\tau}}J_{ij}^{\sigma\tau}\sigma_{i}\tau_{j}+\frac{1}{\sqrt{2N_{\sigma}}}\beta_{\sigma}\sum_{i,j}^{N_{\sigma}}J_{ij}^{\sigma}\sigma_{i}\sigma_{j}+\frac{1}{\sqrt{2N_{\tau}}}\beta_{\tau}\sum_{i,j}^{N_{\tau}}J_{ij}^{\tau}\tau_{i}\tau_{j}\right]+\\
&-\sqrt{1-t}\left[\frac{\sqrt{\beta_{\sigma}'}}{\sqrt{2N_{\sigma}}}\sum_{i,j}^{N_{\sigma}}\hat{J}_{ij}^{\sigma}\sigma_{i}\sigma_{j}+\frac{\sqrt{\beta_{\tau}'}}{\sqrt{2N_{\tau}}}\sum_{i,j}^{N_{\tau}}\hat{J}_{ij}^{\tau}\tau_{i}\tau_{j}\right]+\\
&- \left(\sqrt{x\beta_{\sigma\tau}}\right)\left[\sqrt{a}\sum_{i}J_{i}^{\sigma}\sigma_{i}+\sqrt{a^{-1}}\sum_{i}J_{i}^{\tau}\tau_{i}\right],
\end{alignat*}
where
\begin{equation*}
\begin{array}{ccc}
\beta'_{\sigma}=\beta_{\sigma}^{2}-a^2\alpha\beta_{\sigma\tau}^{2} & ; & \beta_{\tau}'=\beta_{\tau}^{2}-a^{-2}\left(1-\alpha\right)\beta_{\sigma\tau}^{2},
\end{array}
\end{equation*}
$a$ is a positive free parameter and the $\{\mathbf{J}\}$ and $\{\mathbf{\hat{J}}\}$ are all families of unitary gaussian random variable independent from each other. Then we define naturally an interpolating pressure as
\be
A_N(t,x)= -\beta f_N(\beta) =\frac 1 N \mathbb{E}\ln Z_N(t,x)=\frac 1 N \mathbb{E}\ln \sum_{\boldsymbol{\sigma,\tau}}e^{-H_N(t,x)},
\ee
where $f_N(\beta)$ is the standard quenched free energy.
\newline
Finally we define an interpolating action $S_{N}\left(t,x\right)$ that, this time, is not directly the interpolating pressure as in the first part of the paper. Here, we need to add two constants that will be determined \textit{a posteriori}. In other words we define
\begin{equation*}
S_{N}\left(t,x\right)=2A_{N}\left(t,x\right)+Xx+Tt.
\end{equation*}
Deriving the action with respect to $t$ we get
\begin{eqnarray}
\frac{\partial S_{N}\left(t,x\right)}{\partial t}&=&\frac{2}{N}\mathbb{E}\left[Z^{-1}_{N\left(t,x\right)}\sum_{\boldsymbol{\sigma,\tau}}\frac{\partial}{\partial t}e^{-H_{N}\left(t,x\right)}\right]+T\nonumber \\
&=&-\frac{1}{2}\left\langle \left[\beta_{\sigma\tau}\left(\alpha aq_{\sigma\sigma'}+\left(1-\alpha\right)a^{-1}q_{\tau\tau'}\right)\right]^{2}\right\rangle_{(t,x)}+\frac{1}{2}\beta_{\sigma\tau}^{2}(\alpha a +(1-\alpha)a^{-1})+T\nonumber\\
&=&-\frac{1}{2}\left\langle \left[\beta_{\sigma\tau}\left(\alpha aq_{\sigma\sigma'}+\left(1-\alpha\right)a^{-1}q_{\tau\tau'}\right)\right]^{2}\right\rangle_{(t,x)},
\end{eqnarray}
where we have chosen $T=-\frac{1}{2}\beta_{\sigma\tau}^{2}(\alpha a +(1-\alpha)a^{-1})$ in order to have a square product in the last expression. For the derivative with respect to the space variable we have
\begin{eqnarray}
\frac{\partial S_{N}\left(t,x\right)}{\partial x}&=&\frac{2}{N}\mathbb{E}\left[\frac{1}{Z_{N}\left(t,x\right)}\sum_{\left\{ \sigma\right\} }\sum_{\left\{ \tau\right\} }\frac{\partial}{\partial x}e^{-H_{N}\left(t,x\right)}\right]+X\nonumber \\
&=&-\left\langle \beta_{\sigma\tau}\left(\alpha a q_{\sigma\sigma'}+\left(1-\alpha\right)a^{-1}q_{\tau\tau'}\right)\right\rangle_{(t,x)} +\beta_{\sigma\tau}\left(a \alpha+a^{-1}\left(1-\alpha\right)\right)+X\nonumber \\
&=&-\left\langle \beta_{\sigma\tau}\left(\alpha a q_{\sigma\sigma'}+\left(1-\alpha\right)a^{-1}q_{\tau\tau'}\right)\right\rangle_{(t,x)}
\end{eqnarray}
with the choice of $X=-\beta_{\sigma\tau}\left(a \alpha+a^{-1}\left(1-\alpha\right)\right)$. If we call, as in the ferromagnetic case, the velocity field
\be
D_N(t,x)=\partial_xS_N(t,x)=-\beta_{\sigma\tau}\left\langle D_N(\boldsymbol{\sigma,\tau};a)\right\rangle_{(t,x)},
\ee
where we defined the observable $D_N(\boldsymbol{\sigma,\tau};a)=\alpha a q_{\sigma\sigma'}(\boldsymbol{\sigma})+\left(1-\alpha\right)a^{-1}q_{\tau\tau'}(\boldsymbol{\tau})$ that is a linear combination of the two overlaps, we can easily write down an Hamilton-Jacobi equation for $S_N(t,x)$ as
\begin{eqnarray}\label{disHJ}
\partial_tS_N(t,x) &+&\frac 1 2 (\partial_xS_N(t,x))^2 + V_N(t,x)=0, \\
V_N(t,x) &=& -\frac 1 2 \beta_{\sigma\tau}^2\left(\left\langle D_N^2\right\rangle_{(t,x)}-\left\langle D_N\right\rangle_{(t,x)}^2\right)=0.
\end{eqnarray}
In contrast with the ferromagnetic case, where the potential evidently vanished in the thermodynamic limit, in the disordered case the potential  $\left\langle D_N^2\right\rangle_{(t,x)}-\left\langle D_N\right\rangle_{(t,x)}^2$, proportional to the fluctuations of the order parameters, is not in general negligible, neither in the thermodynamic limit \cite{MPV}\cite{Gsumrules}.  Still, if we are looking for a replica-symmetric approximation of the real (full-RSB) solution, we can impose $\lim_{N\to \infty}V_N(t,x)=0$ and try to solve a free Hamilton-Jacobi equation for $S(t,x)$. For this purpose we need to compute first the initial condition for the action\footnote{Here the strength of the method becomes clearly manifest as the calculation of the Cauchy condition $S_N(t=0,x=x_0)$ implies considering only one-body interactions (that trivially factorizes) and whose analytic expression is immediate.}, that is
\begin{eqnarray}
S_{N}\left(0,x\right)&=&\frac{2}{N}\mathbb{E}\ln Z_{N}\left(0,x\right)+Xx=\nonumber\\
&=&\frac{2}{N}\mathbb{E}\ln\sum_{\boldsymbol{\sigma}}e^{\sqrt{\frac{\beta '_{\sigma}}{2N_{\sigma}}}\sum_{i,j}^{N_{\sigma}}\hat{J}_{ij}^{\sigma}\sigma_{i}\sigma_{j}+\sqrt{\beta_{\sigma\tau}a x}\sum_{i}J_{i}^{\sigma}\sigma_{i}} +\frac{2}{N}\mathbb{E}\ln\sum_{\boldsymbol{\tau}}e^{\sqrt{\frac{\beta '_{\tau}}{2N_{\tau}}}\sum_{i,j}^{N_{\tau}}\hat{J}_{ij}^{\tau}\tau_{i}\tau_{j}+\sqrt{\beta_{\sigma\tau}a^{-1}x}\sum_{i}J_{i}^{\tau}\tau_{i}}+ Xx \nonumber
\end{eqnarray}
and contains the free energies of two SK models with external random field and different temperatures $\sqrt{\beta '_{\sigma}}$ and $\sqrt{\beta '_{\tau}}$, i.e.
\be
S_N(0,x)=2\alpha A_N^{SK}(\sqrt{\beta '_{\sigma}},\sqrt{\beta_{\sigma\tau}a x})+2(1-\alpha)A_N^{SK}(\sqrt{\beta '_{\tau}},\sqrt{\beta_{\sigma\tau}a^{-1}x})+Xx.
\ee
Since we are interested in the replica symmetric approximation of the solution, we can use it also in the initial condition, replacing
$A^{SK}(\beta)$ with the well known RS approximation \cite{MPV}
\be
A^{SK}_{RS}(\beta)= \log 2 +\int d\mu(z)\log\cosh(z\sqrt{\beta^2 q})+\frac{\beta^2}{4}(1-q)^2
\ee
with $q=q(\beta)$ solution of the self consistent equation
\be\label{scdisord}
q(\beta)=\int d\mu(z)\tanh ^2(z\sqrt{\beta^2 q (\beta)}).
\ee
As in the ferromagnetic case, taking the derivative with respect to $x$ we get the initial condition for the velocity
\begin{equation}
D(0,x)=\partial_x S(0,x)\nonumber
%&=&-2\beta_{\sigma\tau}\left[\alpha a\int d\mu(z)\tanh^2(z\sqrt{\beta '_{\sigma}q(\sqrt{\beta %'_{\sigma}},\sqrt{\beta_{\sigma\tau}ax})+\beta_{\sigma\tau} ax}) + (1-\alpha) a^{-1}\int d\mu(z)\tanh^2(z\sqrt{\beta '_{\tau}q(\sqrt{\beta %'_{\tau}},\sqrt{\beta_{\sigma\tau}a^{-1}x})+\beta_{\sigma\tau}a^{-1}x})\right]\nonumber \\
=-2\beta_{\sigma\tau}\left[\alpha a\int q(\sqrt{\beta '_{\sigma}},\sqrt{\beta_{\sigma\tau}ax}) + (1-\alpha) a^{-1}q(\sqrt{\beta '_{\tau}},\sqrt{\beta_{\sigma\tau}a^{-1}x})\right],
\end{equation}
that allows us to write the following self consistent equation for $D(t,x)$:
\begin{eqnarray}
D(t,x)&=&D(0,x_{0})=D(0,x-D(t,x)t)=\nonumber \\
&=&-2\beta_{\sigma\tau}\left[\alpha a q(\sqrt{\beta '_{\sigma}},\sqrt{\beta_{\sigma\tau}a(x-D t)}) + (1-\alpha) a^{-1}q(\sqrt{\beta '_{\tau}},\sqrt{\beta_{\sigma\tau}a^{-1}(x-Dt)})\right]
\end{eqnarray}
and finally the solution of the model as
\be\label{dissol}
A(t,x)= \frac 1 2 \left(S(t,x)-Xx-Tt\right)=\frac 1 2 \left(S(0,x-D(t,x)t)+\frac 1 2 D(t,x)^2t -Xx-Tt\right).
\ee
At $t=1$ and $x=0$, when we recover the original model, the velocity field  $D(\boldsymbol{\beta})=D(1,0)$ is the solution of
\be
D(\boldsymbol{\beta})=-2\beta_{\sigma\tau}\left[\alpha a q(\sqrt{\beta '_{\sigma}},\sqrt{-\beta_{\sigma\tau}aD }) + (1-\alpha) a^{-1}q(\sqrt{\beta '_{\tau}},\sqrt{-\beta_{\sigma\tau}a^{-1}D})\right].
\ee
If we call
\begin{eqnarray}
\bar{q}_{\sigma\sigma'}(\boldsymbol{\beta};a)&=&q(\sqrt{\beta '_{\sigma}},\sqrt{-\beta_{\sigma\tau}aD }),\nonumber\\
\bar{q}_{\tau\tau'}(\boldsymbol{\beta};a)&=&q(\sqrt{\beta '_{\tau}},\sqrt{-\beta_{\sigma\tau}aD }),
\end{eqnarray}
since $D(\boldsymbol{\beta})=-2\beta_{\sigma\tau}[\alpha a \bar{q}_{\sigma\sigma'}(\boldsymbol{\beta};a)+(1-\alpha) a^{-1}\bar{q}_{\tau\tau'}(\boldsymbol{\beta};a) ]$ and using the definition of $q(\beta)$, we can easily check that $\bar{q}_{\sigma\sigma'}$ and $\bar{q}_{\tau\tau'}$ satisfy the following system of coupled self-consistent equation, independent by the parameter $a$,
\begin{eqnarray}
\bar{q}_{\sigma\sigma'}&=&\int d\mu(z) \tanh^2(z\sqrt{\beta^2_{\sigma}\bar{q}_{\sigma\sigma'}+(1-\alpha) \beta_{\sigma\tau}^2\bar{q}_{\tau\tau'}}),\nonumber\\
\bar{q}_{\tau\tau'}&=&\int d\mu(z) \tanh^2(z\sqrt{\beta^2_{\tau}\bar{q}_{\tau\tau'}+\alpha \beta_{\sigma\tau}^2\bar{q}_{\sigma\sigma'}}),
\end{eqnarray}
that mirrors exactly what we obtained in the previous section (equation $(\ref{dissc})$). Note that, also in this case, due to the freedom in the choice of the interpolating parameter $a$, i.e.
\be
\alpha a \left\langle q_{\sigma\sigma '}(\sigma)\right\rangle+(1-\alpha) a^{-1} \left\langle q_{\tau\tau'}(\tau)\right\rangle=
\alpha a \bar{q}_{\sigma\sigma'}+(1-\alpha) a^{-1} \bar{q}_{\tau\tau'},
\ee
we can associate $\bar{q}_{\sigma\sigma'}=\left\langle q_{\sigma\sigma '}(\sigma)\right\rangle$ and $\bar{q}_{\tau\tau'}=\left\langle q_{\tau\tau '}(\tau)\right\rangle$ and characterize completely the model. Using the previous decomposition for $D(\boldsymbol{\beta})$ inside the equation ($\ref{dissol}$), we get the free energy of the model in terms of overlaps recovering the main expression enclosed in the statements of Theorem $3$.

%----------------------------------------------------------------------------------------
%		THIRD APPROACH: FOURIER FRAMEWORK - D
%----------------------------------------------------------------------------------------
\subsection{Third approach: the Fourier framework}
Once introduced the mechanical interpolating scheme, we can solve the Hamilton-Jacobi equation ($\ref{disHJ}$) using, as in the ferromagnetic counterpart, the Fourier method too. We can do that again in the replica symmetry approximation in which we neglect the potential, proportional to the fluctuations of the order parameters.
\newline
In this context we have to note that solving a free Hamilton Jacobi equation is equivalent to solving a Burger-like equation
\be
\partial_tS_N(t,x)+\frac 1 2 (\partial_xS_N(t,x))^2+\frac{1}{2N}\partial^2_{x^2}S_N(t,x)=0,
\ee
where we added an irrelevant, because vanishing in the thermodynamic limit, mollifier term proportional to the second derivative of $S_N(t,x)$.
\newline
Hence, also for replica-symmetric bipartite spin-glasses, in this way we can apply a Cole Hopf transform, namely introduce a function $\Psi_N(t,x)$ as
\begin{equation}
\Psi_N(t,x) = \exp\left( -N S_N(t,x) \right).
\end{equation}
Trough the latter we can map the problem of solving for the quenched pressure in statistical mechanics in solving a heat equation for the Cole-Hopf transform of the action, namely
\begin{eqnarray}
\frac{\partial \Psi_N(t,x)}{\partial t} - \frac{1}{2N}\frac{\partial^2 \Psi_N(t,x)}{dt^2}=0
\end{eqnarray}
and follow the prescription of Theorem $\ref{le:Fourier}$ to obtain a variational principle equivalent to the equation ($\ref{dissol}$), hence solving the Fourier equation in the impulse space and, due to the monotonicity of the exponential, reverse the expression for the action as
\begin{equation}
S_N(t,x) = -\frac1N \ln \Psi_N(t,x) = -\frac1N \ln \sqrt{\frac{N}{2\pi t}}\int dy \exp\left( N ( S_0(y) + \frac{(x-y)^2}{2t}  )\right),
\end{equation}
that, in the thermodynamic limit, returns the solution as the inverse Legendre transform of the initial condition
\begin{equation}
S(t,x) = \inf_y \left(  \frac{(x-y)^2}{2t} +S_0(y) \right).
\end{equation}

%----------------------------------------------------------------------------------------
%
%		CONCLUSIONS AND OUTLOOKS
%
%----------------------------------------------------------------------------------------
\section{Conclusions and Outlooks}
In this paper we have shown how to adapt techniques originally stemmed mainly in the classical mechanics scenario in order to make them powerful tools for solving the statistical mechanics of mean field spin systems too, focusing on bipartite structures in full interaction. In a sense this work extends, merges (and closes, at least at the replica symmetric level), our investigations started in \cite{bipCWSK} we and \cite{Fourier} on mean field spin systems in interaction. In particular, in this paper we considered the test case of two parties, each one provided of its internal links and in reciprocal interaction with the other party: we investigated both the ferromagnetic case, where parties share the positivity of the couplings (whose strength is instead tunable in each party and reciprocally) and the glassy counterpart, where, retaining the same freedom in the strengths, couplings are drawn at random from a Gaussian distribution allowing for positive and negative strengths, hence frustrating the network.
\newline
At first we proved that it is possible to built a sum rule for the free energy (strictly speaking the pressure) of these models in terms of a replica symmetric expression plus a rest that is exactly the source of order parameter fluctuations, then, if these order parameters are self-averaging (as in the ferromagnetic case or in the replica symmetric approximations), such an expression becomes the true solution in the thermodynamic limit. We stress that, however, for glassy systems, in  a huge region of the tunable parameters (namely where the rest in the sum rule is positive defined) such an expression is further a rigorous bound for the real free energy.
\newline
If self-averaging is lacking, instead, as in the low temperature limit of glassy systems, the expression for the free energy is only an approximation. We remark however that in several applicative fields (e.g. ranging from neural to immune or metabolite networks in theoretical biology) this level of description is retained, hence motivating the present study.
\newline
One step forward, we showed that there exists a sharp one to one mapping between the free energy of these systems in the statistical mechanics scenario and an action function in a suitably defined fictitious spacetime such that solving the latter implies solving the former: following this path, we have shown how to obtain an explicit expression (again at the replica symmetric level) for the action and then map back this finding in the original statistical mechanics framework reobtaining the same solutions (both for ferromagnets and for glasses) previously discovered.
\newline
Lastly, we have shown that the Cole-Hopf transform of the free energy obeys a diffusion-like equation that we solved via the standard route of Green propagator and convolution theorem in the impulse space and then we mapped it back in the original frame, re-obtaining once more the same thermodynamics.
\newline
As a final remark, we stress here that extensions of these techniques to a (finite in number) amount of different species (beyond the test-case of two groups investigated here) is straightforward.
\newline
Summarizing, we believe that, while the self-averaging scenario is completely  understood, from multiple perspectives, and rules out further investigations on ferromagnets with multi-species, the phenomenon of replica symmetry breaking in multiple spin-glasses still deserves much more efforts for being tackled.
\newline
We plan to investigate its structure in the near future.

\section*{Acknowledgements }

This work is supported by the FIRB grant RBFR08EKEV. Further we thank Sapienza Universita' di Roma, Istituto Nazionale di Fisica Nucleare (INFN) and Gruppo Nazionale per la Fisica Matematica (GNFM, INdAM) for partial support.

%----------------------------------------------------------------------------------------
%
%	APPENDICES
%
%----------------------------------------------------------------------------------------

%\newpage

%\addtocontents{toc}{\vspace{2em}} % Add a gap in the Contents, for aesthetics

%\appendix % Cue to tell LaTeX that the following 'chapters' are Appendices
%
%% Include the appendices of the thesis as separate files from the Appendices folder
%% Uncomment the lines as you write the Appendices
%
%\input{./Appendices/usefulformulas}

\addtocontents{toc}{\vspace{2em}} % Add a gap in the Contents, for aesthetics

%\backmatter

%----------------------------------------------------------------------------------------
%
%	BIBLIOGRAPHY
%
%----------------------------------------------------------------------------------------

%\newpage

%\addcontentsline{toc}{chapter}{\bibname}
%\bibliographystyle{abbrv}
%\bibliography{Bibliography}

%\addcontentsline{toc}{chapter}{\bibname}
%\bibliographystyle{abbrv}
%\bibliography{Bibliography}

%\label{Bibliography}
%
%%\lhead{\emph{Bibliography}} % Change the page header to say "Bibliography"
%\bibliographystyle{plain}
%%\bibliographystyle{unsrtnat} % Use the "unsrtnat" BibTeX style for formatting the Bibliography
%%\bibliographystyle{plain_ita}
%\bibliography{Bibliography}% The references (bibliography) information are stored in the file named "Bibliography.bib"

\end{document}